\documentclass[12pt]{article}%[a4paper]{article}
\usepackage[latin1]{inputenc}
\usepackage{graphicx}
\usepackage{epsfig}
\usepackage{color}
\usepackage{lscape}
\usepackage{verbatim}
\usepackage{amsthm, amssymb}
\usepackage{amsmath}
\newtheorem{Lem}{Lemma}
\newtheorem{theorem}{Theorem}

\title{Min $st$-Cut of a Planar Graph in $O(n\log\log n)$ Time}
\author{Christian Wulff-Nilsen
        \footnote{School of Computer Science,
                  Carleton University,
                  \texttt{koolooz@diku.dk},
                  \texttt{http://cg.scs.carleton.ca/$_{\widetilde{~}}$cwn/}.
                  Research partially supported by NSERC and MRI.}}

\begin{document}

\maketitle
\begin{abstract}
Given a planar undirected $n$-vertex graph $G$ with non-negative edge weights, we show how to compute,
for given vertices $s$ and $t$ in $G$, a min $st$-cut in $G$ in $O(n\log\log n)$ time and $O(n)$ space. The previous best
time bound was $O(n\log n)$.
\end{abstract}

\section{Introduction}
Given a graph $G = (V,E)$ with non-negative edge weights and given vertices $s,t\in V$, an \emph{$st$-cut}
of $G$ is a partition of $V$ into two subsets $S$ and $T = V\setminus S$ such that $s\in S$ and $t\in T$. The
\emph{weight} of the $st$-cut is the sum of weights of edges starting in $S$ and ending in $T$. A \emph{min $st$-cut} of $G$
is an $st$-cut of $G$ having minimum weight.

Computing a min $st$-cut is a classical algorithmic problem with several applications in areas such as chip design, communication
networks, transportation, and cluster analysis. The problem is intimately related to another well-studied problem, that of computing
a max $st$-flow. The classical max flow min cut theorem implies that the weight of a min $st$-cut is the value of a max $st$-flow.

For general graphs, the fastest known max $st$-flow algorithm runs in time $O(mn\log(m^2/n))$, where $m$ is the number of edges
and $n$ is the number of vertices~\cite{MaxFlowGeneral}. Another implication of the max flow min cut theorem is that a min
$st$-cut can be obtained from a max $st$-flow in linear time. Hence, a min $st$-cut can also be computed in $O(mn\log(m^2/n))$ time.

For planar undirected graphs, Reif~\cite{Reif} showed how to solve the min $st$-cut problem in $O(n\log^2n)$ time. This was later
improved to $O(n\log n)$ by Frederickson~\cite{APSPPlanar}. For directed planar graphs, Borradaile and Klein~\cite{MinCutPlanar}
showed that the max $st$-flow problem can be solved in $O(n\log n)$ time and this gives an $O(n\log n)$ time min $st$-cut algorithm
for such graphs.

In this paper, we give a min $st$-cut algorithm for planar undirected graph with $O(n\log\log n)$ running time and $O(n)$ space
requirement, thereby
improving the time bound by Frederickson~\cite{APSPPlanar}. In order to achieve this, we do not depart from Reif's
approach~\cite{Reif}. Instead we speed it up using a two-phase approach. The first phase runs a ``coarse'' version of
Reif's algorithm which only determines a subset of the min $st$-cut candidates found by the original algorithm. We obtain a
running time of $O(n\log\log n)$ for this phase using the fast Dijkstra variant of Fakcharoenphol and Rao~\cite{Fakcharoenphol}.
In the second phase, the remaining min $st$-cut candidates are found exactly as in the algorithm by Reif but since the first phase
partitions the problem into simpler subproblems, we can show that the second phase also runs in $O(n\log\log n)$ time.

The organization of the paper is as follows. In Section~\ref{sec:Preliminaries}, we give some definitions and introduce some of the
tools that we need. We briefly go through the ideas of Reif's algorithm in Section~\ref{sec:Reif} before presenting
our algorithm in Section~\ref{sec:FasterAlgo}. One step of our algorithm constructs a certain division of the graph and we present
the details of this step in Section~\ref{sec:rDiv}. Finally, we make some concluding remarks and suggestions for future research
in Section~\ref{sec:ConclRem}.

\section{Preliminaries}\label{sec:Preliminaries}
For a graph $G = (V,E)$, define a \emph{piece} $P = (V_P,E_P)$ of $G$ to be the subgraph of $G$ defined by a subset
$E_P$ of $E$. In $G$, the vertices of $V_P$ incident to vertices in $V\setminus V_P$ are the \emph{boundary vertices} of $P$.
Vertices of $V_P$ that are not boundary vertices of $P$ are \emph{interior vertices} of $P$. If $G$ is edge-weighted,
we define the \emph{dense distance graph} of $P$ to be the complete graph on the set of boundary vertices of $P$ where
each edge $(u,v)$ has weight equal to the shortest path distance (w.r.t.\ the edge weights) in $P$ between $u$ and $v$.

Let $G = (V,E)$ be an $n$-vertex planar graph with a non-negative weight function $w:V\rightarrow\mathbb R$ defined on its vertices.
For a subset $A$ of $V$, define $w(A) = \sum_{v\in A}w(v)$. We assume that $w(V) = 1$. The separator theorem
of Lipton and Tarjan states that in $O(n)$ time, $V$ can be partitioned into three subsets $A$, $B$, and $C$
such that
\begin{itemize}
\item no edge joins a vertex in $A$ with a vertex in $B$,
\item $\frac 1 3\leq w(A)\leq \frac 2 3$ and $\frac 1 3\leq w(B)\leq\frac 2 3$, and
\item $|C| = O(\sqrt n)$.
\end{itemize}

Using this theorem, Frederickson~\cite{APSPPlanar} showed how to obtain, for any parameter $r\in (0,n)$, an
\emph{$r$-division} of $G$, which is a division of (the edges of) $G$ into $O(n/r)$ pieces each containing $O(r)$ vertices and
$O(\sqrt r)$ boundary vertices. He gave an $O(n\log r + (n/\sqrt r)\log n)$ time algorithm to find such a division.

A stronger version of the separator theorem is the cycle separator theorem of Miller~\cite{CycleSep} which states that if $G$ is
a \emph{plane} graph then $C$ can be chosen such that there exists a Jordan curve that only intersects $G$ in
vertices of $C$. Miller showed that such a separator can be found in linear time.

We will show that by applying the cycle separator theorem as well as ideas of Fakcharoenphol and
Rao~\cite{Fakcharoenphol}, we can obtain the $r$-division of Frederickson but with some additional properties. More precisely, define
the \emph{holes} of a piece to be the internal faces containing boundary vertices. We prove the following result in
Section~\ref{sec:rDiv}.
\begin{theorem}\label{Thm:rDiv}
For a plane $n$-vertex graph, an $r$-division in which each piece has $O(1)$ holes can be found in $O(n\log r + (n/\sqrt r)\log n)$
time.
\end{theorem}
In the following, when we talk about an $r$-division, we shall assume that it has the form in Theorem~\ref{Thm:rDiv}.

We shall identify an $st$-cut with the set of edges from the $s$-side to the $t$-side of the cut.

\section{Reif's Algorithm}\label{sec:Reif}
Reif's algorithm~\cite{Reif} makes use of the following duality between cuts in the primal graph and cycles in the dual graph:
a min $st$-cut in a plane undirected graph $G$ corresponds to a minimum weight simple cycle separating faces $s$ and $t$ in
the dual $G^\ast$ of $G$; here, a simple cycle is said to separate two faces if one face is in the interior and the other is in the
exterior of the cycle.

In the first step of Reif's algorithm, a shortest path
$P = p_1\rightarrow p_2\rightarrow\cdots\rightarrow p_{|P|}$ from an arbitrary vertex $p_1$ on face $s$ to an arbitrary vertex
$p_{|P|}$ on face $t$ in $G^\ast$ is computed.
%Let us choose the embedding of $G^\ast$ such that $t$ is the external face.
Then $G^\ast$ is ``cut open'' along
$P$ as follows. Remove the set $E_r$ of edges emanating right of $P$ in the direction from $s$ to $t$. Insert a copy
$P' = p_1'\rightarrow p_2'\rightarrow\cdots\rightarrow p_{|P|}'$ of $P$ and for each edge $(p_i,u)\in E_r$, add edge $(p_i',u)$.
We let $G_{st}^\ast$ be the resulting graph, see Figure~\ref{fig:Reif}(a).
\begin{figure}
\centerline{\scalebox{0.6}{\input{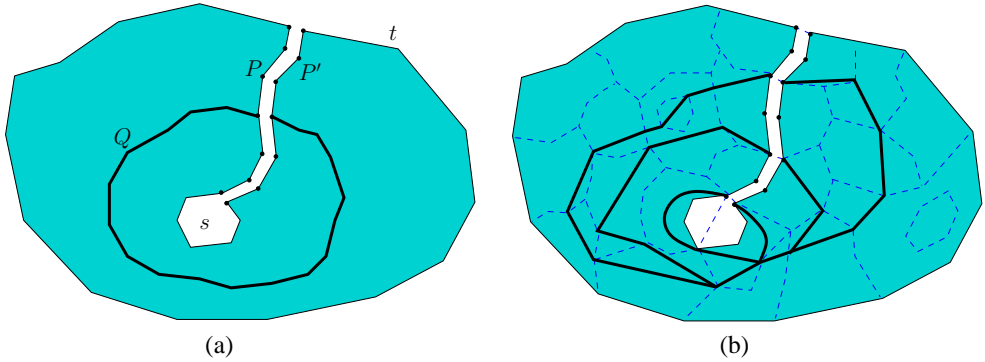}}}
\caption{(a): In cut-open graph $G_{st}^\ast$, Reif's algorithm computes a shortest path $Q$ from the midpoint on $P$ to the
         midpoint on $P'$ and recurses on the two subgraphs generated. (b): The coarse version of Reif's algorithm only
         computes shortest paths between boundary vertices on the cut-path. A refined version is then applied to find
         the remaining shortest paths. Only shortest paths from the coarse version are shown. Dashed line segments show
         the boundaries of pieces in the $r$-division.}
\label{fig:Reif}
\end{figure}

Next, Reif's algorithm computes a shortest path $Q$ in $G_{st}^\ast$ from the midpoint $p_{\lceil|P|/2\rceil}$ of $P$ to the
midpoint $p_{\lceil|P|/2\rceil}'$ of $P'$. This splits
$G_{st}^\ast$ into two
subgraphs and splits $P$ and $P'$ into two halves, one for each side of $Q$. In each of the two subgraphs, degree two-vertices
are removed by merging their incident edges. The algorithm then recurses on the two subgraphs and the two subpaths.

Let $Q_i$ be the shortest path found by the algorithm and let $p_i$ and $p_i'$ be the first and last vertex of $Q_i$, respectively.
Then the cycle in $G^\ast$ obtained from $Q_i$ by identifying $p_i$ with $p_i'$ is a minimum-weight $st$-separating cycle in
$G^\ast$. By the min cycle/min cut duality, this cycle defines a min $st$-cut in primal graph $G$.

With Dijkstra's shortest path algorithm, Reif's algorithm runs in $O(n\log^2n)$ time. This can be improved to $O(n\log n)$ time
with Frederickson's algorithm~\cite{APSPPlanar} or by speeding up Reif's algorithm using the linear time shortest path algorithm
of Henzinger et al.~\cite{SSSPPlanar}. In the next section, we will further improve Reif's algorithm to get $O(n\log\log n)$
running time.

\section{The $O(n\log\log n)$ Time Algorithm}\label{sec:FasterAlgo}
In this section, we present our algorithm and give the claimed $O(n\log\log n)$ time bound. To ease the presentation, we
leave out some details of the algorithm and return to them in Section~\ref{sec:CycleOverlap}.

We start with the following simple lemma.
\begin{Lem}\label{Lem:FastReif}
Let $s$ and $t$ be faces in a planar $n$-vertex graph $G$ and let $P$ be a given shortest path between a vertex on
$s$ and a vertex on $t$. In an application of Reif's algorithm to find a minimum weight $st$-separating cycle in $G$,
consider a subproblem defined by a subgraph $H$ and a subpath of $P$ of length $O(\log^cn)$ for a constant $c$. Then this
subproblem can be solved in $O(|H|\log\log n)$ time.
\end{Lem}
\begin{proof}
Recursion depth for Reif's algorithm in $H$ is only $O(\log(\log^cn)) = O(\log\log n)$ so the running time for the subproblem
is $O(|H|\log\log n)$ using the linear time shortest path algorithm in~\cite{SSSPPlanar}.
\end{proof}

We essentially run Reif's algorithm but speed part of it up with the Dijkstra variant of Fakcharoenphol and
Rao~\cite{Fakcharoenphol}. In the following, let $G$ denote the dual of a plane embedding of the input graph. We need to find a
minimum weight $st$-separating cycle in $G$, where $s$ and $t$ are faces.

Let $P$ be a shortest path in $G$ from an arbitrary vertex on $s$ to an arbitrary vertex on $t$. We can find this path in linear
time using the algorithm in~\cite{SSSPPlanar}. We first run a ``coarse'' version of Reif. This will identify in $O(n\log\log n)$ time
a subset of all the $st$-separating cycles found by the original algorithm. More precisely, the
set of cycles found will split $G$ into subgraphs each of which contains a subpath of $P$ of length $O(\log^cn)$ for some constant
$c$. We then run the ``refined'' Reif algorithm by applying Lemma~\ref{Lem:FastReif} to each subgraph and its associated
subpath of $P$. This will find the minimum weight $st$-separating cycle in $G$. By ensuring that the total size of the subgraphs is
$O(n)$, the entire algorithm runs in $O(n\log\log n)$ time. Figure~\ref{fig:Reif}(b) illustrates the output of the first phase of our
algorithm.

\subsection{First phase}
We will now describe the first phase of our algorithm which is the coarse version of Reif's algorithm.

\paragraph{$r$-division}
First, we apply Theorem~\ref{Thm:rDiv} to obtain an $r$-division of $G$ for
$r = \log^6n$. This takes $O(n\log r + (n/\sqrt r)\log n) = O(n\log\log n)$ time.

\paragraph{Cutting pieces open}
We make an incision in $G$ along shortest path $P$ as in Reif's algorithm. This induces incisions in those pieces containing
parts of $P$ and we update the pieces accordingly. If a boundary vertex of a piece belongs to $P$ before the incision, we regard both
of its two copies after the incision as boundary vertices of that piece. Note that there will still be only $O(\sqrt r)$ boundary
vertices in each piece and these boundary vertices will still be on a constant number of faces after the incision. Hence, the
resulting set of pieces forms an $r$-division in the cut-open graph.

\paragraph{Dense distance graphs}
Next, we compute dense distance graphs of the pieces in the $r$-division. To do this, we shall apply Klein's multiple-source shortest
paths algorithm~\cite{MultiSrc}. For an $h$-vertex plane graph $H$ and a fixed face $f$ of $H$, this algorithm builds a data
structure in $O(h\log h)$ time and space such that shortest path queries between vertex pairs $(u,v)$, where either $u$ or $v$ is
on $f$, can be answered in $O(\log h)$ time per query.

For each piece, we apply Klein's algorithm to set up a data structure for the external face and query this data structure for
shortest path distances in the piece from boundary vertices on this face to all other boundary vertices. Since there are
$O(r)$ pairs of boundary vertices, we obtain all these distances in $O(r\log r)$ time. We similarly set up data structures for
each hole and get the distances between the remaining pairs of boundary vertices. Since the piece has a constant number of holes,
total time to construct its dense distance graph is $O(r\log r)$. Over all pieces, this is
$O((n/r)r\log r) = O(n\log r) = O(n\log\log n)$ time. We represent the edge weights of each dense distance graph in a distance
matrix with $O(\sqrt r)$ rows and columns.

\paragraph{Fast Dijkstra}
Fakcharoenphol and Rao~\cite{Fakcharoenphol} gave an efficient Dijkstra variant for planar graphs. More precisely,
they showed the following. Given a collection of pieces each having a
constant number of holes and given their dense distance graphs, any shortest path tree in the union of these dense distance
graphs can be computed in $O(h\log^2n)$ time, where $h$ is the total number of \emph{vertices} in these graphs. In general,
this time bound is sublinear in the number of edges.

Since we have computed dense distance graphs for the pieces in our $r$-division and since the total number of boundary vertices
of these pieces is $O(n/\sqrt r)$, it follows that a shortest path between any two boundary vertices in the $r$-division can
be computed in $O((n/\sqrt r)\log^2n) = O(n/\log n)$ time. Note that this shortest path consists of edges from the dense distance graphs
so it is an implicit representation of a shortest path in the underlying cut-open graph.

\paragraph{Coarse Reif}
Let $p_1,\ldots,p_{|P|}$ be the ordered sequence of vertices of $P$, starting with the vertex on face $s$. Let
$p_{i_1},\ldots,p_{i_k}$ be the (possibly empty) subsequence of vertices that are boundary vertices in pieces of the $r$-division.
Note that these vertices partition $P$ into subpaths each of which is contained in a piece.

The coarse version of Reif's algorithm is the normal algorithm of Reif restricted to subsequence $p_{i_1},\ldots,p_{i_k}$
and using the $O(n/\log n)$ time shortest path algorithm for each vertex in the subsequence.
Since recursion depth is $O(\log n)$, total time for the first phase of our algorithm is $O(n)$ in addition to the
$O(n\log\log n)$ time to find the $r$-division and to set up the dense distance graphs. The $st$-separating cycles found in this
phase partition $G$ into subgraphs each containing a subpath of $P$ fully contained in a piece of
the $r$-division. Hence, the length of each such subpath is bounded by the size $O(r) = O(\log^6n)$ of a piece.

\paragraph{Subgraphs for recursive calls}%\label{sec:SplitPieces}
%When cutting $G$ open along $P$,
%we replace each boundary vertex in the $r$-division by two new boundary vertices, one for each of the two copies
%of $P$ in the cut-open graph. The dense distance graph for each cut-open piece inherits the distances from the original
%dense distance graph except that distances between pairs of boundary vertices belonging to different components of the
%cut-open piece are set to infinite in the distance matrix. We can still apply the Dijkstra variant of~\cite{Fakcharoenphol}
%in the cut-open piece since the number of faces containing boundary vertices remains bounded by a constant. For a pair of
%boundary vertices in the union of cut-open pieces, this algorithm will then give us the weight of a shortest path in $G$ between the
%corresponding vertices in $G$ among those paths that do not cross $P$. Running time for the Dijkstra variant is
%still $O((n/\sqrt r)\log^2n) = O(n/\log n)$ since the number of boundary vertices increases by a factor of at most two
%when cutting pieces open.
When applying the coarse version of Reif's algorithm, we need to find the subgraphs for recursive calls.
Consider a subgraph $H$ in some recursive call. We associate with $H$ the boundary vertices belonging
to $H$ and the cyclic orderings of these vertices on holes and external faces of pieces. When running the fast Dijkstra
variant for $H$, only distances between boundary vertices associated with $H$ are considered in the dense distance graphs. Hence,
these graphs need not be updated in recursive calls. The time to find a shortest path in $H$ will be $O(h\log^2n)$, where $h$ is
the number of boundary vertices in $H$.

\subsection{Second phase}
In order to run the second phase of our algorithm, we need to convert the shortest paths consisting of edges from dense distance
graphs to the underlying shortest paths in $G$ and we need to find the subgraphs of $G$ bounded by these paths.
In Section~\ref{sec:CycleOverlap} we show how to do this in $O(n)$ time such that the total size of the subgraphs is $O(n)$.
Applying Lemma~\ref{Lem:FastReif} with constant $c = 6$ to each subgraph, we get $O(n\log\log n)$ time for the second phase of
our algorithm. Hence, the entire algorithm has $O(n\log\log n)$ running time.

\subsection{Overlapping subgraphs}\label{sec:CycleOverlap}
We need to ensure that the total size of the subgraphs generated by the coarse version
of Reif's algorithm is not too large, i.e., we need the total size to be linear. The problem is that subgraphs overlap so a
vertex can belong to several subgraphs. The original algorithm of Reif ensures linear total size by deleting, in every subgraph
generated, each degree two vertex by replacing the two edges $e_1$ and $e_2$ incident to it by one whose weight is the sum of the
weights of $e_1$ and $e_2$.

\paragraph{First phase}
In the first phase of our algorithm, we do something similar: if a boundary vertex $p$ is incident (in the union of dense distance
graphs) to only two other boundary
vertices $q_1$ and $q_2$ for the current subgraph, $p$ is removed and a single \emph{super edge} is added between $q_1$ and $q_2$ whose
weight is equal to the sum of weights of edges $(q_1,p)$ and $(p,q_2)$. We repeat this process until no boundary vertices of degree two
exist. This will ensure that the total number of boundary vertices
over all subgraphs generated is $O(n/\sqrt r)$ and this will also be a bound on the total number of super edges generated.

The shortest path algorithm of~\cite{Fakcharoenphol} can easily be extended to deal with super edges in addition to the dense
distance graphs: simply regard each super edge as a dense distance graph consisting of two vertices and one edge. Hence,
the total running time for computing shortest paths in the first phase is $O((n/\sqrt r)\log^3n) = O(n)$.

\paragraph{Second phase}
We also face the problem with overlapping subgraphs when converting the shortest paths consisting of dense distance graph edges
to shortest paths in $G$ for the second phase of the algorithm.
Klein's algorithm~\cite{MultiSrc} can report the underlying path in $G$ corresponding to a dense distance
graph edge in time proportional to the length of the path (this requires that $G$ has constant degree which we can assume without
loss of generality). Hence, after the first phase we can obtain each of the shortest paths computed in time proportional to their total
size. However, this size can be super-linear since the paths can share many vertices.
We deal with this problem in the following. We will show that an implicit representation of the paths can be computed in $O(n)$ time.

\paragraph{Implicit representation of paths}
Define $p_{i_1},\ldots,p_{i_k}$ as above, i.e., the ordered sequence of vertices of $P$
from which the coarse Reif algorithm has computed shortest paths $P_{i_1},\ldots,P_{i_k}$. Let $p_{i_1}',\ldots,p_{i_k}'$ be the
endpoints of these paths. We start by obtaining the shortest
path $Q_{i_1}$ in $G$ from $p_{i_1}$ to $p_{i_1}'$ using Klein's algorithm on each dense distance graph edge of $P_{i_1}$.
This takes $O(|Q_{i_1}|)$ time.

To find the shortest path $Q_{i_2}$ in $G$ from $p_{i_2}$ to $p_{i_2}'$, we similarly apply Klein's algorithm on $P_{i_2}$. If we
encounter no vertices already visited, we obtain the entire path and move on to $P_{i_3}$. Otherwise, let $v_1$ be the first
already visited vertex and let $Q_{v_1}$ be the path found. We stop the algorithm when reaching $v_1$ and instead start
obtaining vertices of $Q_{i_2}$ backwards from $p_{i_2}'$ until reaching an already visited vertex $v_2$. Let $Q_{v_2}$ be the path
found but ordered from $v_2$ to $p_{i_2}'$.

If $v_2$ is on $Q_{i_1}$, a simple property of shortest paths
allows us to choose $Q_{i_2}$ as the concatenation of $Q_{v_1}$, the subpath of $Q_{i_1}$ from $v_1$ to $v_2$, and $Q_{v_2}$, in
that order. This gives us an implicit representation of $Q_{i_2}$ in $O(|Q_{i_2}\setminus Q_{i_1}|)$ time.

In a dense distance graph, an edge $(u,v)$ need not represent the same underlying shortest path as the edge $(v,u)$ since shortest
paths need not be unique. Hence, it may happen that $v_2$ is on $Q_{v_1}$ and not on $Q_{i_1}$.
If so, we can redefine $Q_{i_2}$ to be the subpath of $Q_{v_1}$ from $p_{i_2}$ to $v_2$ followed by $Q_{v_2}$.
Letting $Q_{v_2,v_1}$ be the subpath of $Q_{v_1}$ from $v_2$ to $v_1$, total time to find $Q_{i_2}$ is
$O(|Q_{i_2}| + |Q_{v_2,v_1}|) = O(|Q_{i_2}\setminus Q_{i_1}| + |Q_{v_2,v_1}|)$.
Since none of the vertices on $Q_{v_2,v_1}$, excluding $v_2$, will be visited again, we can afford to spend time $O(|Q_{v_2,v_1}|)$.

Repeating this process for the remaining shortest paths $P_{i_3},\ldots,P_{i_k}$ gives an implicit representation of the
corresponding shortest paths in $G$ and it follows from the above analysis that running time is $O(n)$. In linear time it is then
easy to obtain from this representation the desired subgraphs needed in the second phase of our algorithm and to ensure that
they have total linear size.

\section{$r$-division}\label{sec:rDiv}
In this section, we prove Theorem~\ref{Thm:rDiv}, i.e., we show that an $r$-division can be found in $O(n\log r + (n/\sqrt r)\log n)$
time. We use an approach similar to that of Frederickson~\cite{APSPPlanar}: contract $G$, find an $r$-division of this smaller
graph, expand the graph back to $G$, and split some of the resulting pieces further to get the desired $r$-division of the whole graph.
First, however, we shall give a simple $O(n\log n)$ time algorithm. In Section~\ref{subsec:FastrDiv}, this algorithm will be
used to find an $r$-division of the contracted graph.

\subsection{Weak $r$-division}
To obtain an $r$-division of $G$ in $O(n\log n)$ time, we again follow Frederickson's approach. First we find a
\emph{weak $r$-division} which is a division of $G$ into
$O(n/r)$ pieces each of size $O(r)$ and with a constant number of holes,
such that the total number of boundary vertices over all pieces is $O(n/\sqrt r)$.
In Section~\ref{subsec:rDiv}, we will then split pieces further to get the desired $r$-division in $O(n\log n)$ time.

Consider the following recursive algorithm to find a weak $r$-division: regard $G$ as a piece with no boundary vertices, split it
recursively into two subpieces with the cycle separator theorem of Miller and recursive on them. The recursion stops when a piece
has size at most $r$. As shown by Frederickson~\cite{APSPPlanar}, this gives a weak $r$-division. However, it does not ensure a
constant bound on the number of holes in each piece which we need in our application.

To deal with this, we use ideas of Fakcharoenphol and Rao~\cite{Fakcharoenphol} to keep the number of holes bounded by some constant
$h$. The initial piece is the whole graph and thus contains no holes. Now, consider the general recursive step and let $P = (V_P,E_P)$
be the current piece. Assume it has at most $h$ holes. Apply the cycle separator theorem to $P$ with all vertices assigned weight
$1/|V_P|$. This splits $P$ into two subpieces $P_1 = (V_{P_1},E_{P_1})$ and $P_1' = (V_{P_1'},E_{P_1'})$, where
$|V_{P_1}| = \alpha|V_P| + O(\sqrt{|V_P|})$ and $|V_{P_1'}| = (1-\alpha)|V_P| + O(\sqrt{|V_P|})$ for some
$\frac 1 3\leq\alpha\leq \frac 2 3$.
Assume w.l.o.g.\ that $P_1$ belongs to the interior of the separator cycle. Then this subpiece has at most $h$ holes. However,
since the separator cycle may have introduced a new hole in $P_1'$, this subpiece may have $h + 1$ holes.

Contract the holes of $P_1'$
into \emph{super vertices} and apply the cycle separator theorem with vertex weights distributed evenly on super vertices. Expand them
back to holes and let $P_2$ and $P_3$ be the resulting two pieces. As shown in~\cite{Fakcharoenphol}, the number of holes in
each of the two subpieces will be a constant factor smaller than $h + 1$ so if we pick $h$ sufficiently large, $P_2$ and $P_3$ each
have at most $h$ holes. Now, we recurse on $P_1$, $P_2$, and $P_3$ until all pieces contain at most $r$ vertices.
\begin{Lem}\label{Lem:WeakrDiv}
The above procedure gives, for any parameter $r\in(0,n)$, a weak $r$-division of $G$ where each piece has a constant number of holes.
Running time is $O(n\log (n/r))$.
\end{Lem}
\begin{proof}
We have already argued that the number of holes in the pieces generated is constant.
We now show that the total number of boundary vertices over all pieces is $O(n/\sqrt r)$.

For any boundary vertex $v$ in the
weak $r$-division, let $b(v)$ denote one less than the number of pieces containing $v$. Let $B(n)$ be the sum of $b(v)$ over all
such $v$. There are nonnegative values $\alpha_1$, $\alpha_2$, and $\alpha_3$ with $\alpha_1 + \alpha_2 + \alpha_3 = 1$ such that
in the above procedure, $P_i$ contains at most $\alpha_in + c'\sqrt n$ vertices, $i = 1,2,3$. Note that
$\frac 1 3\leq\alpha_1\leq \frac 2 3$. We shall assume that $\alpha_2\geq\alpha_3$. For $n > r$,
\[
  B(n) \leq c\sqrt n + B(\alpha_1n + c'\sqrt n) + B(\alpha_2n + c'\sqrt n) + B(\alpha_3n + c'\sqrt n)
\]
for constants $c,c' > 0$, and $B(n) = 0$ for $n\leq r$. We will prove by induction on $n\geq \frac r 9$ that
$B(n) \leq d(n/\sqrt r - \frac 1 3\sqrt n)$ for some constant $d > 0$ (to be specified).

Clearly, this holds for $\frac r 9\leq n\leq r$ for any choice of $d > 0$ so assume that $n > r$ and that the claim holds for
smaller values. We have $\alpha_1\geq \frac 1 3$ and since $\alpha_2\geq\alpha_3$ and $\alpha_2 + \alpha_3\geq \frac 1 3$, we have
$\alpha_2\geq \frac 1 6$. Hence, both
$\alpha_1n + c'\sqrt n$ and $\alpha_2n + c'\sqrt n$ are at least $\frac n 6 > \frac r 9$ so the induction hypothesis can be applied to
both of these values. We distinguish between two cases: $\alpha_3n + c'\sqrt n\geq \frac r 9$ and $\alpha_3n + c'\sqrt n < \frac r 9$.
Assume first that $\alpha_3n + c'\sqrt n\geq \frac r 9$. The induction hypothesis gives
\begin{align*}
  B(n,r) & \leq c\sqrt n + \frac{dn}{\sqrt r} + \frac{3dc'\sqrt n}{\sqrt r} - {}\\
         & \phantom{\leq {}} \frac d 3\left(\sqrt{\alpha_1n + c'\sqrt n} + \sqrt{\alpha_2n + c'\sqrt n} +
                                            \sqrt{\alpha_3n + c'\sqrt n}\right).
\end{align*}

We will prove that the right-hand side is at most $d(n/\sqrt r - \frac 1 3\sqrt n)$, i.e., that
\[
  c + \frac{3dc'}{\sqrt r} \leq \frac d 3\left(\sqrt{\alpha_1 + \frac{c'}{\sqrt n}} + \sqrt{\alpha_2 + \frac{c'}{\sqrt n}} +
                                               \sqrt{\alpha_3 + \frac{c'}{\sqrt n}} - 1\right),
\]
which will follow if we can show that
\begin{equation}
  1 + \frac{3c}d + \frac{9c'}{\sqrt r} \leq \sqrt{\alpha_1} + \sqrt{\alpha_2} + \sqrt{\alpha_3}.\label{alpha_eqn}
\end{equation}

We may assume that $r$ is at least some large constant. Picking $d$ sufficiently large as well, we can make the left-hand
side in~(\ref{alpha_eqn}) equal to $1 + \epsilon$ for an arbitrarily small constant $\epsilon > 0$. Since $\alpha_2,\alpha_3\leq 1$
and since $\alpha_1 + \alpha_2 + \alpha_3 = 1$,
\[
  \sqrt{\alpha_1} + \sqrt{\alpha_2} + \sqrt{\alpha_3} \geq
  (\sqrt{\alpha_1} - \alpha_1) + \alpha_1 + \alpha_2 + \alpha_3 =
  \sqrt{\alpha_1} - \alpha_1 + 1.
\]
Since $\frac 1 3\leq\alpha_1\leq \frac 2 3$, the right-hand side in~(\ref{alpha_eqn}) is larger
than $\sqrt{\frac 2 3} - \frac 2 3 + 1 > 1$. This proves the induction step for the case $\alpha_3n + c'\sqrt n\geq \frac r 9$.

Now, assume that $\alpha_3n + c'\sqrt n < \frac r 9$. Then $B(\alpha_3n + c'\sqrt n) = 0$ and the induction hypothesis gives
\[
  B(n,r) \leq c\sqrt n + \frac{dn}{\sqrt r} + \frac{2dc'\sqrt n}{\sqrt r} -
              \frac d 3\left(\sqrt{\alpha_1n + c'\sqrt n} + \sqrt{\alpha_2n + c'\sqrt n}\right).
\]
The induction step will follow from the inequality
\begin{equation}
  1 + \frac{3c}d + \frac{6c'}{\sqrt r} \leq \sqrt{\alpha_1} + \sqrt{\alpha_2}.\label{alpha_eqn2}
\end{equation}
Since
\[
  \alpha_3n < \alpha_3n + c'\sqrt n < \frac r 9 < \frac n 9,
\]
we have $\alpha_3 < \frac 1 9$ and hence $\alpha_1 + \alpha_2 > \frac 8 9$. Since also $\frac 1 3\leq\alpha_1\leq \frac 2 3$,
the right-hand side of~(\ref{alpha_eqn2}) is greater than
\[
  (\sqrt{\alpha_1} - \alpha_1) + \alpha_1 + \alpha_2 > \sqrt{\frac 2 3} - \frac 2 3  + \frac 8 9 > 1
\]
and~(\ref{alpha_eqn2}) follows by picking $d$ and $r$ sufficiently large.

We have shown that the total number of boundary vertices over all pieces is $O(n/\sqrt r)$. To show that the procedure generates a
weak $r$-division, we also need to give an $O(n/r)$ bound on the number of pieces. Pieces of the form $P_1$ or $P_2$ each have
size $\Theta(r)$. Since the total number of vertices over all pieces is $n + B(n) = O(n)$, the number of such pieces is $O(n/r)$.
The number of pieces of the form $P_3$ cannot be larger than the number of the form $P_1$ (or $P_2$). Hence, the total number
of pieces is $O(n/r)$.

It follows that the procedure generates a weak $r$-division. Since Miller's cycle separator can be found in linear time and since
the procedure recurses until pieces have size at most $r$, running time is $O(n\log(n/r))$.
\end{proof}

\subsection{$r$-division in $O(n\log n)$ time}\label{subsec:rDiv}
To obtain an $r$-division in $O(n\log n)$ time, we first find a weak $r$-division with Lemma~\ref{Lem:WeakrDiv}. Each piece has
$O(r)$ vertices and a constant number of holes but there may be more than order $\sqrt r$ boundary vertices in the piece.

We continue to follow Frederickson's approach while ensuring a constant number of holes in each piece. If there is a piece
$P$ containing more than $c\sqrt r$ boundary vertices for some constant $c$, apply the cycle separator theorem as in the
weak $r$-division procedure to obtain subpieces $P_1$ and $P_1'$. However, instead of distributing the vertex weights evenly on all
vertices of $P$ when
applying the theorem, we now distribute weights evenly on boundary vertices only. We then infer subpieces $P_2$ and $P_3$ of $P_1'$
as before by distributing vertex weights evenly on super vertices defined by contracted holes of $P_1'$.
This is repeated until each piece has at most $c\sqrt r$ boundary vertices.
\begin{Lem}\label{Lem:rDiv}
The above procedure gives, for any parameter $r\in(0,n)$, an $r$-division of $G$ in $O(n\log n)$ time.
\end{Lem}
\begin{proof}
The proof is more or less identical to that in~\cite{APSPPlanar}. We include it here for completeness.
In the weak $r$-division, let $t_i$ be the number of pieces with exactly $i$ boundary vertices. From the proof of
Lemma~\ref{Lem:WeakrDiv}, we have $\sum_i it_i = \sum_{v\in V_B}(b(v) + 1)$, where $V_B$ is the set of boundary vertices over all
pieces in the weak $r$-division. Hence, $\sum_{v\in V_B}(b(v) + 1) < 2B(n) = O(n/\sqrt r)$.

In the weak $r$-division, consider a piece $P$ with $i > c\sqrt r$ boundary vertices. When the above procedure splits $P$ into
subpieces $P_1$, $P_2$, and $P_3$, each of them contains at most a constant fraction of the boundary vertices of $P$. Hence, after
$di/(c\sqrt r)$ splits of $P$ for some constant $d$, all subpieces will contain at most $c\sqrt r$ boundary vertices. This will result
in at most $1 + di/(c\sqrt r)$ subpieces and at most $c'\sqrt r$ new boundary vertices per split for some constant $c'$. We may assume
that $c'\leq c$. The total number of new boundary vertices introduced by the above procedure is thus
\[
  \sum_i (c'\sqrt r)(di/(c\sqrt r))t_i\leq d\sum_i it_i = O(n/\sqrt r)
\]
and the number of new pieces is at most
\[
  \sum_i (di/(c\sqrt r))t_i = O(n/r).
\]
Hence, the procedure generates an $r$-division. Since a weak $r$-division can be found in $O(n\log n)$ time, an $r$-division can also
be found within this time bound.
\end{proof}

\subsection{A faster algorithm}\label{subsec:FastrDiv}
We now show how to get the desired running time of $O(n\log r + (n/\sqrt r)\log n)$ in Theorem~\ref{Thm:rDiv}.
We start by computing a spanning tree $T$ of $G$ (here, we assume that $G$ is connected; we can always add
infinite-weight edges to achieve this) and partitioning it into $\Theta(n/\sqrt r)$ subtrees each of size $\Theta(\sqrt r)$; the
subtrees cover all vertices and are pairwise vertex-disjoint. This takes $O(n)$ time with the algorithm in~\cite{Fre1}.

Let $G'$ be a plane multigraph obtained from $G$ by contracting each subtree to a single vertex.
This graph contains $O(n/\sqrt r)$ vertices. To obtain the same asymptotic bound on the number of edges, we will turn $G'$ into a
so called thin graph. In a plane multigraph, a \emph{bigon} is a face defined by two vertices and edges.
A plane multigraph is \emph{thin} if it contains no bigons. The following result from~\cite{ThinGraph} shows that thin multigraphs
are sparse.
\begin{Lem}\label{Lem:Thin}
A thin $n$-vertex multigraph contains $O(n)$ edges.
\end{Lem}
Let $G''$ be the thin multigraph obtained from $G'$ by identifying the two edges of a bigon with one edge and repeating this
process until no bigons exist. Graph $G''$ is turned into a simple graph $G'''$ by subdividing each edge $(u,v)$ into two edges
$(u,w)$ and $(w,v)$. Note that $(u,v)$ corresponds to $O(\sqrt r)$ edges in $G'$ and in $G$; we subdivide each of them similarly
such that the sum of weights of each edge pair equals the weight of the edge they subdivide. Now, $G'$ is a simple graph and we
colour black those vertices of $G'$ that correspond to contracted trees in $G$. All other vertices of $G'$ are coloured white.

By Lemma~\ref{Lem:Thin}, $G'''$ is a simple planar graph of size $O(n/\sqrt r)$ so we can find an $r$-division of it in
$O((n/\sqrt r)\log n)$ time with Lemma~\ref{Lem:rDiv}. This $r$-division consists of $O(n/r^{3/2})$ pieces each of
size $O(r)$. We get an induced division of $G'$ into pieces each consisting of $O(r)$ black vertices and $O(r^{3/2})$ white
vertices. Furthermore, each piece in this division has $O(\sqrt r)$ black boundary vertices and $O(r)$ white boundary vertices.

Let $\mathcal P_1$ be the set of subtrees of $G$ defined by the expanded black boundary vertices of pieces in the division of $G'$.
Note that $|\mathcal P_1| = O(n/r)$ and each subgraph in $\mathcal P_1$ has size $O(\sqrt r)$.

For each piece $P$ in the division of $G'$, let $P'$ be the piece in $G$ defined by the union of edges in $P$ and
subtrees from expanded black interior vertices of $P$. Let $\mathcal P_2$ denote the set of these pieces $P'$. Note
that $|\mathcal P_2| = \Theta(n/r^{3/2})$ and each piece $P'$ in $\mathcal P_2$ has size $O(r^{3/2})$. Furthermore, since there are
$O(r)$ white boundary vertices in $P'$ and each of the $O(\sqrt r)$ black boundary vertices of the corresponding piece
in $G'$ contributes with at most $O(\sqrt r)$ boundary vertices to $P'$ when expanded, $P'$ has $O(r)$ boundary vertices.

The pieces in $\mathcal P_1\cup\mathcal P_2$ cover all edges in $G$ and each edge is contained in exactly one piece.
We will transform these pieces into an $r$-division of $G$.

First consider pieces $P\in\mathcal P_1$. The number of boundary vertices of $P$ is bounded by the size $O(\sqrt r)$ of $P$.
Since $P$ is a tree, all its boundary vertices are on the external face. Hence, $P$ has no holes and we include it as part of
the $r$-division of $G$.

Now, consider pieces $P\in\mathcal P_2$. The piece $P'$ in the division of $G'$ corresponding to $P$ has a constant number
of holes. We claim that the same holds for $P$. For consider some black boundary vertex $v$ in $P'$ and let $T_v$ be the tree
in $G$ obtained by expanding $v$. In $P$, $v$ gets expanded into $O(\sqrt r)$ boundary vertices all belonging to $T_v$. Since
none of the edges of $T_v$ belong to $P$ by definition, these boundary vertices must all be on the same face of $P$; see
Figure~\ref{fig:Expand}.
\begin{figure}
\centerline{\scalebox{0.6}{\input{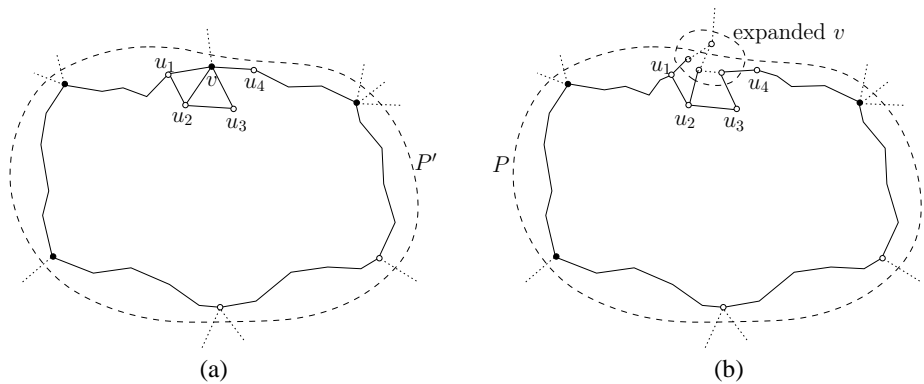}}}
\caption{(a): A piece $P'$ in the $r$-division of $G'$ with a black boundary vertex $v$ on the external face. (b): After expanding
         $v$ to a tree when forming a corresponding piece $P\in\mathcal P_2$, new boundary vertices will also be on the
         external face. The same is true for holes. Only solid edges are part of the pieces.}
\label{fig:Expand}
\end{figure}
Repeating this argument
for all black boundary vertices of $P'$, it follows that $P$ has a constant number of holes.

Since $P$ has size $O(r^{3/2})$ and $O(r)$ boundary vertices, we can find an $r$-division of it in $O(r^{3/2}\log r)$ time using
Lemma~\ref{Lem:rDiv} with a small modification: when finding a weak $r$-division of $P$, the boundary vertices and the holes of
$P$ will be regarded as boundary vertices and holes in the initial graph $P$. The result will still be a weak $r$-division of $P$
since the total number of boundary vertices will be $O(|P|/\sqrt r + r) = O(r) = O(|P|/\sqrt r)$.

Total time to find $r$-divisions over all $P\in\mathcal P_2$ is $O(n\log r)$. Taking the
union of the pieces obtained in all these $r$-divisions together with the pieces in $\mathcal P_1$, we obtain the $r$-division
of $G$ in $O(n\log r + (n/\sqrt r)\log n)$ time. This proves Theorem~\ref{Thm:rDiv}.

\section{Concluding Remarks}\label{sec:ConclRem}
We showed how to compute a min $st$-cut of a planar undirected $n$-vertex graph in $O(n\log\log n)$ time, improving on an earlier
$O(n\log n)$ bound. Can we get linear running time? Does a matching time bound also hold for planar directed graphs and for the
maximum $st$-flow problem in planar (directed or undirected) graphs?

The multiple source shortest path algorithm of Klein~\cite{MultiSrc} has been generalized to bounded genus graphs~\cite{MultiSrcGenus}.
Fakcharoenphol and Rao~\cite{Fakcharoenphol} mention that their algorithm might apply to such graphs as well. It is therefore
natural to conjecture that our time bound also holds for bounded genus graphs.

\section*{Acknowledgments}
I wish to thank Sergio Cabello for his comments and remarks and for making corrections to an earlier version of this paper.

\end{document}